\documentclass[a4paper]{article}
\usepackage{microtype}
\usepackage{a4wide}
\usepackage{enumerate}
\usepackage{tikz}
\usepackage{xspace}
\usetikzlibrary{decorations.pathreplacing,decorations.pathmorphing,patterns,calc,arrows}
\usepackage{url}

\usepackage{float,todonotes,eucal}
\usepackage{amsmath,amsfonts,amssymb,epsfig,color,amsthm}

\usepackage{thmtools, thm-restate}
\newtheorem{theorem}{Theorem}
\newtheorem{corollary}[theorem]{Corollary}
\newtheorem{lemma}[theorem]{Lemma}

\newtheorem{claim}{Claim}
\newtheorem{invariant}{Invariant}
\usepackage[absolute]{textpos}

\usepackage{xcolor}
\definecolor{ForestGreen}{rgb}{0.1333,0.5451,0.1333}
\usepackage[
  pagebackref,
  colorlinks,
  citecolor=ForestGreen,
]{hyperref}

\newcommand{\ceil}[1]{\ensuremath{\left\lceil{#1}\right\rceil}}%

\newcommand{\ProblemFormat}[1]{{\sc #1}}
\newcommand{\ProblemName}[1]{\ProblemFormat{#1}\xspace}

\newcommand{\probLECMul}{\ProblemName{List Edge Coloring in Multigraphs}}
\newcommand{\probLECSim}{\ProblemName{List Edge Coloring in Simple Graphs}}
\newcommand{\probEC}{\ProblemName{Edge Coloring}}
\newcommand{\probLEC}{\ProblemName{List Edge Coloring}}

\newcommand{\probTSAT}{\ProblemName{$3$-SAT}}
\newcommand{\probTFSAT}{\ProblemName{$(3,\!4)$-SAT}}

\newcommand{\vrb}{{\ensuremath{\rm{vrb}}}}
\newcommand{\cls}{{\ensuremath{\rm{cls}}}}

\newcommand{\Cc}{{\ensuremath{\mathcal{C}}}}

\title{Tight Lower Bounds for List Edge Coloring\footnote{Work supported by the National Science Centre of Poland, grant number 2015/17/N/ST6/01224 (AS). The work of {\L}. Kowalik is a part of the project TOTAL that has received funding from the European Research Council (ERC) under the European Union’s Horizon 2020 research and innovation programme (grant agreement No 677651). }} 

\author{\L ukasz Kowalik\thanks{Institute of Informatics, Faculty of Mathematics, Informatics and Mechanics, University of Warsaw, Poland.}  \and Arkadiusz Soca\l a\footnotemark[2]}
\date{}

\begin{document}

\maketitle
\begin{textblock}{20}(0.7, 12.98)
	\includegraphics[width=40px]{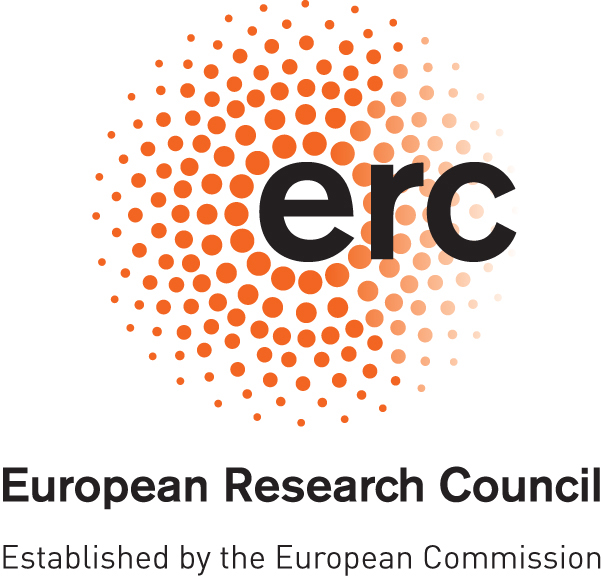}%
\end{textblock}
\begin{textblock}{20}(0.46, 13.28)
	\includegraphics[width=60px]{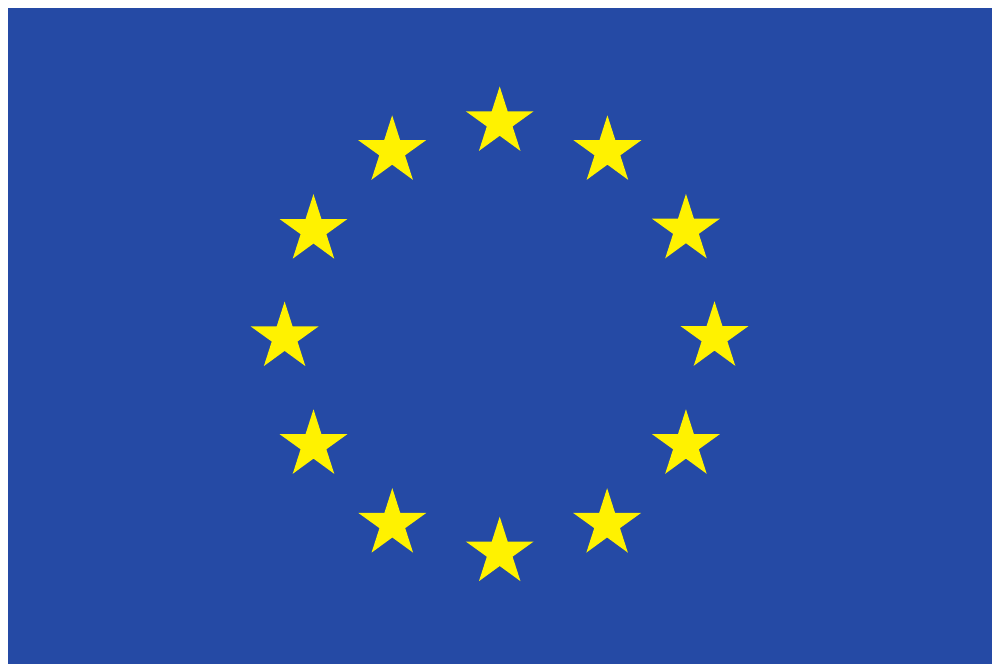}%
\end{textblock}

\begin{abstract}
The fastest algorithms for edge coloring run in time $2^m n^{O(1)}$, where $m$ and $n$ are the number of edges and vertices of the input graph, respectively.
For dense graphs, this bound becomes $2^{\Theta(n^2)}$.
This is a somewhat unique situation, since most of the studied graph problems admit algorithms running in time $2^{O(n\log n)}$.
It is a notorious open problem to either show an algorithm for edge coloring running in time $2^{o(n^2)}$ or to refute it, assuming Exponential Time Hypothesis (ETH) or other well established assumption.

We notice that the same question can be asked for list edge coloring, a well-studied generalization of edge coloring where every edge comes with a set (often called a {\em list}) of allowed colors.
Our main result states that list edge coloring for simple graphs does not admit an algorithm running in time $2^{o(n^2)}$, unless ETH fails. 
Interestingly, the algorithm for edge coloring running in time $2^m n^{O(1)}$ generalizes to the list version without any asymptotic slow-down.
Thus, our lower bound is essentially tight. 
This also means that in order to design an algorithm running in time $2^{o(n^2)}$ for edge coloring, one has to exploit its special features compared to the list version.
\end{abstract}

\section{Introduction}\label{sect:intro}

An edge coloring of a graph $G=(V,E)$ is a function $c:E\rightarrow\mathbb{N}$ which has different values (called colors) on incident edges.
This is one of the most basic graph concepts with plethora of results, including classical theorems of Vizing, Shannon and K\H{o}nig.
In the decision problem \probEC we are given a simple graph $G$ and an integer $k$.
We ask if $G$ can be edge colored using only $k$ colors.
This is an NP-complete problem, as shown by Holyer~\cite{holyer}, similarly as many other natural graph decision problems like {\sc Clique}, {\sc Vertex Coloring}, {\sc Hamiltonicity} or {\sc Subgraph Isomorphism}.
However, there is an intriguing difference between our understanding of \probEC and most of the studied graph problems, including the four ones mentioned above.
Namely, the latter ones admit algorithms running in time $2^{O(n\log n)}$, and often even $2^{O(n)}$ for an $n$-vertex input graph, while it is not known whether \probEC can be solved in time $2^{o(n^2)}$.
Indeed, the fastest known algorithm for edge coloring is obtained by applying the vertex coloring algorithm of Bj\"{o}rklund, Husfeldt and Koivisto~\cite{BjorklundHK09} to the line graph of the input graph. As a result, we get an edge coloring algorithm which, for any graph with $m$ edges and $n$ vertices, runs in time $2^{m}n^{O(1)}$ and exponential space, which is $2^{\Theta(n^2)}$ for dense graphs.
The only progress towards a tailor-made approach for edge coloring is the more recent algorithm of Bj\"{o}rklund, Husfeldt, Kaski and Koivisto~\cite{bhkk:narrow} which still runs in time  $2^{m}n^{O(1)}$ but uses only polynomial space.
In this context it is natural to ask for a lower bound.
Clearly, any superpolynomial lower bound would imply P$\ne$NP.
However, a more feasible goal is to prove a meaningful lower bound under the assumption of a well established conjecture, like Exponential Time Hypothesis (ETH, see Section~\ref{sect:pre} for a precise formulation).
The reduction of Holyer, combined with standard tools (see Section~\ref{sect:pre}) proves that \probEC does not admit an algorithm in time $2^{o(m)}$ or $2^{o(n)}$ .
At the open problem session of Dagstuhl Seminar 08431 in 2008~\cite{dagstuhl08} it was asked to exclude $2^{O(n)}$ algorithms, assuming ETH.
Despite considerable progress in ETH-based lower bounds in recent years~\cite{CyganFGKMPS17,CyganPP16,marx-beat} this problem stays unsolved~\cite{dagstuhl17}.

List edge coloring is a generalization of edge coloring. 
An \emph{edge list assignment} $L:E(G)\rightarrow 2^\mathbb{N}$ is a function that assigns to each edge $e$ of $G$ a set (often called a \emph{list}) $L(e)$ of allowed colors.
A function $c:E(G) \rightarrow \mathbb{N}$ is a \emph{list edge coloring} of $(G,L)$ if $c(e) \in L(e)$ for every $e \in E(G)$, and $c(e) \neq c(f)$ for every pair of incident edges $e, f \in E(G)$.
The notion of list edge coloring is also a frequent topic of research. 
For example, it is conjectured that if $G$ can be edge colored in $k$ colors for some $k$, then it can be list edge colored for any edge list assignment with all lists of size at least $k$.
This conjecture has been proved in some classes of graphs like bipartite graphs~\cite{galvin} or planar graphs of maximum degree at least 12~\cite{boro:edge-list}.

In this work, we study the computational complexity of list edge coloring.
The basic decision problem, \probLECSim, asks if for a given simple graph $G$ with edge list assignment $L$ there is a list edge coloring of $(G,L)$.
Its more general variant, called \probLECMul asks the same question but the input graph does not need to be simple, i.e., it can contain parallel edges.
Although the problem seems much more general than \probEC, the two best known algorithms~\cite{BjorklundHK09,bhkk:narrow} that decide if a given graph admits an edge coloring in $k$ colors solve \probLECMul (and hence also \probLECSim) within the same time bound, i.e.,
$2^{m}m^{O(1)} + O(L)$, where $L$ is the total length of all lists, 
after only minor modifications (see Proposition 3 in~\cite{BjorklundHK09}).
Multigraphs do not admit any upper bound on the number of edges, hence this time complexity does not translate to a function on $n$.
We show that this is not an accident, because satisfiability of any sufficiently sparse 3-CNF-SAT formula can be efficiently encoded as a list edge coloring instance with a bounded number of vertices.
This gives the following result.

\begin{restatable}{theorem}{thmulti}
\label{th:multi}%
If there is a function $f:\mathbb{N}\rightarrow\mathbb{N}$ such that \probLECMul can be solved in time~$f(n)\cdot m^{O(1)}$ for any input graph on $n$ vertices and $m$ edges, then $P=NP$.
\end{restatable}

For simple graphs $m=O(n^2)$ and hence \probLECSim admits an algorithm running in time $2^{O(n^2)}$.
Our main result states that this bound is essentially optimal, assuming ETH.

\begin{restatable}{theorem}{thsimple}
\label{th:simple}%
If there is an algorithm for \probLECSim that runs in time~$2^{o(n^2)}$, then Exponential Time Hypothesis fails.
\end{restatable}

Our results have twofold consequences for the \probEC problem.
First, one may hope that our reductions can inspire a reduction for \probEC.
However, it is possible that such a reduction does not exist and researchers may still try to get an algorithm for \probEC running in time $2^{o(n^2)}$.
Then we offer a simple way of verifying if a new idea works: if it applies to the list version as well, there is no hope for it.

\section{Preliminaries}\label{sect:pre}
For an integer $k$, we denote $[k]=\{0,\ldots,k-1\}$.
If $I$ and $J$ are instances of decision problems $P$ and $R$, respectively, then we say that $I$ and $J$ are {\em equivalent}
if either both $I$ and $J$ are YES-instances of respective problems, or both are NO-instances.
A clause in a CNF-formula is represented by the set of its literals.
For two subsets of vertices $A$, $B$ of a graph $G=(V,E)$ by $E(A,B)$ we denote the set of edges with one endpoint in $A$ and the other in $B$.

\paragraph*{Exponential-Time Hypothesis.}
The Exponential Time Hypothesis (ETH) of Impagliazzo et al.~\cite{eth} states that there exists a constant $c > 0$, such that there is no algorithm solving \probTSAT in time $O(2^{cn})$.
During the recent years, ETH became the central conjecture used for proving tight bounds on the complexity of various problems. 
One of the most important results connected to ETH is the {\em{Sparsification Lemma}}~\cite{seth}, which essentially gives a (many-one) reduction from an arbitrary instance of {\sc{$k$-SAT}} to an instance where
the number of clauses is linear in the number of variables.
The following well-known corollary can be derived by combining ETH with the Sparsification Lemma.

\begin{theorem}[see e.g.~Theorem 14.4 in \cite{platypus}]
\label{thm:eth-main}
Unless ETH fails, there is no algorithm for \probTSAT that runs in time $2^{o(n+m)}$, where $n,m$ denote the numbers of variables and clauses, respectively.
\end{theorem}

We need the following regularization result of Tovey~\cite{Tovey84}.
Following Tovey, by \probTFSAT we call the variant of \probTSAT where each clause of the input formula contains exactly $3$ different variables, and each variable occurs in at most $4$ clauses.

\begin{lemma}[\cite{Tovey84}]
	\label{lem:tovey}
	Given a \probTSAT formula $\varphi$ with $n$ variables and $m$ clauses one can transform
	it in polynomial time into an equivalent \probTFSAT instance $\varphi'$ with $O(n+m)$ 
	variables and clauses.
\end{lemma}


\begin{corollary}
 \label{cor:eth-3,4-sat}
Unless ETH fails, there is no algorithm for \probTFSAT that runs in time $2^{o(n)}$, where $n$ denotes the number of variables of the input formula. 
\end{corollary}

\section{Hardness of \probLEC in Multigraphs}\label{sect:multi}

In order to prove Theorems~\ref{th:multi} and~\ref{th:simple} we show reductions from \probTFSAT to \probLEC with strong bounds on the number of vertices in the output instance.
The basic idea of both our reductions is to use two colors, denoted by $x_i$ and $\neg x_i$ for every variable $x_i$ so that in every coloring of the out graph the edges colored in $x_i$ or $\neg x_i$ form a single path with alternating colors. 
Then colors at the edges of this path of fixed parity can encode the value of $x_i$ in a satisfying boolean assignment. 
Moreover, testing a clause $C=\ell_1\vee \ell_2\vee\ell_3$ can be done very easily: it suffices to add an edge with the list $\{\ell_1,\ell_2,\ell_3\}$.
However this edge can belong to the alternating path of at most one of the three variables in $C$, and we add two more parallel edges which become elements of the two other alternating paths.
Unfortunately, in order to get similar phenomenon in simple graphs, we need to introduce a complicated gadget.

\begin{lemma}
\label{lem:reduction-multi}
For any instance $\varphi$ of \probTFSAT with $n$ variables there is an equivalent instance $(G,L)$ of \probLECMul with 21 vertices and $O(n)$ edges.
Moreover, the instance $(G,L)$ can be constructed in polynomial time.
\end{lemma}

In what follows, we prove Lemma~\ref{lem:reduction-multi}.
Let $\vrb(\varphi)$ and $\cls(\varphi)$ be the sets of variables and clauses of $\varphi$, respectively.
W.l.o.g.\ assume $\vrb(\varphi)=\{x_0,\ldots,x_{n-1}\}$.

We construct an auxiliary graph $G_\varphi$ with $V(G_\varphi)=\cls(\varphi)$ and such that two clauses $C_1,C_2\in \cls(\varphi)$ are adjacent in $G_\varphi$ iff $C_1\cap C_2 \ne \emptyset$.
Since every clause has three variables and each variable can belong to at most three other clauses, it follows that the maximum degree of $G_\varphi$ is at most $9$.
Let $g:\cls(\varphi)\rightarrow[10]$ be the greedy vertex coloring of $G_\varphi$ in 10 colors, which can be found in linear time in a standard way.
For $i\in [10]$, let $\Cc_i=g^{-1}(i)$.

Let us describe the output instance $(G,L)$.
We put $V(G)=\{v_0,\ldots,v_{20}\}$.
The edges of $G$ join only vertices of consecutive indices. 
For every $r\in [10]$, for every clause $C\in \Cc_r$ we add three new edges with endpoints $v_{2r}$ and $v_{2r+1}$.
The first of this edges, denoted by $e_C^1$, gets list $C$, i.e., the three literals of clause $C$.
Let $x_i$, $x_j$ and $x_k$ be the three variables that appear in $C$.
Then, the two remaining edges, $e_C^2$ and $e_C^3$, get identical lists of $\{x_i,\neg x_i,x_j,\neg x_j,x_k,\neg x_k\}$.
Moreover, for every $r\in [10]$ and for every variable $x_i$ that does not appear in any of the clauses of $\Cc_r$, we add a new edge $v_{2r}v_{2r+1}$ with list $\{x_i,\neg x_i\}$.
Finally, for every $r\in [10]$ and for every variable $x_i\in \vrb(\varphi)$ we add a single new edge $v_{2r+1}v_{2r+2}$ with list $\{x_i,\neg x_i\}$.
This finishes the description of the output instance. See Fig.~\ref{fig:ex-multi} for an example.

  \begin{figure}
	\begin{center}
		\begin{tikzpicture}[scale=0.72, transform shape]
		\tikzset{vtx/.style={draw, circle, line width=1pt, inner sep = 1pt}}
		\tikzset{inv/.style={circle, line width=1pt, inner sep = 1pt}}
		\tikzset{edgecolor/.style={draw, rectangle, line width=.5pt, fill=white, inner sep=2pt, above}}

		\node [vtx] (v0) at (0, 2) {$v_0$};
		\node [vtx] (v1) at (5,2) {$v_1$};
		\node [vtx] (v2) at (10,2) {$v_2$};
		\node [vtx] (v3) at (15,2) {$v_3$};
		
		\draw [thick] (v0) -- +(0,1cm) edge [out=80,in=100] node [above] {$x_1, x_2, \neg x_3$} ([yshift=1cm]v1);
		\draw [thick] (v1) -- +(0,1cm);
		\draw [thick] (v0) edge [bend left=50] node [above,pos=0.5] {$x_1, \neg x_1, x_2, \neg x_2, x_3, \neg x_3$} (v1);
		\draw [thick] (v0) edge node [above,pos=0.5] {$x_1, \neg x_1, x_2, \neg x_2, x_3, \neg x_3$} (v1);
		\draw [thick] (v0) edge [bend right=50] node [above,pos=0.5] {$x_4,\neg x_4$} (v1);
		\draw [thick] (v0) -- +(0,-1cm) edge [out=-80,in=-100] node [above] {$x_5, \neg x_5$} ([yshift=-1cm]v1);
		\draw [thick] (v1) -- +(0,-1cm);
		
		\draw [thick] (v1) -- +(0,1cm) edge [out=80,in=100] node [above] {$x_1, \neg x_1$} ([yshift=1cm]v2);
		\draw [thick] (v2) -- +(0,1cm);
		\draw [thick] (v1) edge [bend left=50] node [above,pos=0.5] {$x_2,\neg x_2$} (v2);
		\draw [thick] (v1) edge node [above,pos=0.5] {$x_3,\neg x_3$} (v2);
		\draw [thick] (v1) edge [bend right=50] node [above,pos=0.5] {$x_4,\neg x_4$} (v2);
		\draw [thick] (v1) -- +(0,-1cm) edge [out=-80,in=-100] node [above] {$x_5, \neg x_5$} ([yshift=-1cm]v2);
		\draw [thick] (v2) -- +(0,-1cm);
		
		\draw [thick] (v2) -- +(0,1cm) edge [out=80,in=100] node [above] {$\neg x_2, \neg x_4, \neg x_5$} ([yshift=1cm]v3);
		\draw [thick] (v3) -- +(0,1cm);
		\draw [thick] (v2) edge [bend left=50] node [above,pos=0.5] {$x_2, \neg x_2, x_4, \neg x_4, x_5, \neg x_5$} (v3);
		\draw [thick] (v2) edge node [above,pos=0.5] {$x_2, \neg x_2, x_4, \neg x_4, x_5, \neg x_5$} (v3);
		\draw [thick] (v2) edge [bend right=50] node [above,pos=0.5] {$x_1,\neg x_1$} (v3);
		\draw [thick] (v2) -- +(0,-1cm) edge [out=-80,in=-100] node [above] {$x_3, \neg x_3$} ([yshift=-1cm]v3);
		\draw [thick] (v3) -- +(0,-1cm);
		

		\end{tikzpicture}
	\end{center}
	\caption{\label{fig:ex-multi}Edges related to clauses $(x_1 \vee x_2 \vee \neg x_3)$ and $(\neg x_2 \vee \neg x_4 \vee \neg x_5)$ assuming that the first of these clauses has color $0$ and the second has color $1$.}
\end{figure}
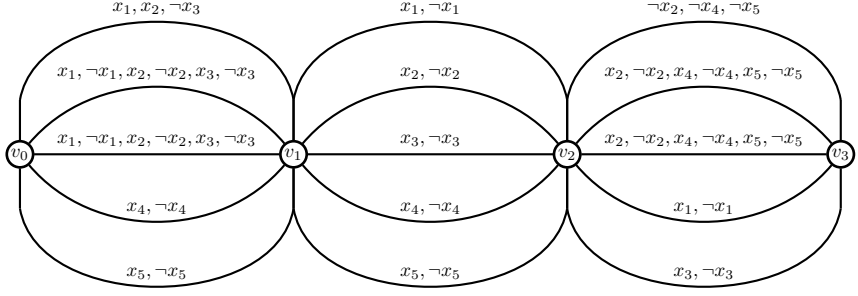

In what follows, edges of the form $v_{2r}v_{2r+1}$ are called {\em positive} and edges of the form $v_{2r+1}v_{2r+2}$ are called {\em negative}.

   \begin{claim} \label{claim:variable-path}
    For every list edge coloring $c$ of $(G,L)$, for every $i\in [n]$, the edges in $c^{-1}(\{x_i,\neg x_i\})$ form a path $v_0,v_1,\ldots,v_{20}$.
   \end{claim}

   \begin{proof}
   For every $r\in[10]$, there is exactly one edge $v_{2r+1}v_{2r+2}$ with list containing $x_i$ or $\neg x_i$, namely with list $\{x_i,\neg x_i\}$.
   It follows that these 10 edges belong to $c^{-1}(\{x_i,\neg x_i\})$.
   It suffices to prove that for every $r\in[10]$ there is also exactly one edge $v_{2r}v_{2r+1}$ in $c^{-1}(\{x_i,\neg x_i\})$.
   This is clear when $x_i$  does not appear in any of the clauses of $\Cc_r$, because then there is exactly one edge $v_{2r+1}v_{2r+2}$ with list containing $x_i$ or $\neg x_i$, namely with list $\{x_i,\neg x_i\}$.
   Otherwise, let $C=\{\ell_i,\ell_j,\ell_k\}$ be the clause of $\Cc_r$ where $\ell_i\in\{x_i,\neg x_i\}$.
   Let $\ell_j\in\{x_j,\neg x_j\}$, $\ell_k\in\{x_k,\neg x_k\}$.
   Then there are exactly three edges $e_C^1$, $e_C^2$, $e_C^3$ incident to $v_{2r}$ and $v_{2r+1}$ and with list containing one of literals in the set $\{x_i,\neg x_i,x_j,\neg x_j,x_k,\neg x_k\}$.
   Indeed, $L(e_C^1)=\{\ell_i,\ell_j,\ell_k\}$, and $L(e_C^2)=L(e_C^3)=\{x_i,\neg x_i,x_j,\neg x_j,x_k,\neg x_k\}$. 
   However, we have already proved that for every $q\in\{i,j,k\}$, one of the edges with endpoints $v_{2r+1}$ and $v_{2r+2}$ is colored with $x_q$ or $\neg x_q$.
   Hence, since every color class is a matching, for every $q\in\{i,j,k\}$, at most one of the edges in $\{e_C^1,e_C^2,e_C^3\}$ is colored with $x_q$ or $\neg x_q$. 
   However, lists of $e_C^1$, $e_C^2$, $e_C^3$ contain only colors of the form $x_q$ or $\neg x_q$ for $q\in\{i,j,k\}$.
   It follows that for every $q\in\{i,j,k\}$ exactly one of the edges in $\{e_C^1,e_C^2,e_C^3\}$ is colored with $x_q$ or $\neg x_q$. 
   In particular there is exactly one edge $v_{2r}v_{2r+1}$ in $c^{-1}(\{x_i,\neg x_i\})$. 
   \end{proof}

   Since $c$ is an edge coloring, the path from the claim above is colored either by $x_i,\neg x_i, x_i, \neg x_i,\ldots$, or by $\neg x_i, x_i, \neg x_i,x_i,\ldots$.
   This implies the following claim.
   
   \begin{claim} \label{claim:one-value-multi}
    For every list edge coloring $c$ of $(G,L)$, for every $i\in [n]$, we have $|c^{-1}(x_i)|=|c^{-1}(\neg x_i)|=10$ and either all edges in $c^{-1}(x_i)$ are positive and all edges in $c^{-1}(\neg x_i)$ are negative or all edges in $c^{-1}(x_i)$ are negative and all edges in $c^{-1}(\neg x_i)$ are positive.
   \end{claim}

   Now we are ready to prove that $\varphi$ and $(G,L)$ are equivalent.
   
   Assume $c$ is a list edge coloring of $(G,L)$. Define a boolean assignment $f:\vrb(\varphi)\rightarrow\{T,F\}$ by setting $x_i$ to $T$ iff all edges in $c^{-1}(x_i)$ are positive. 
   Now consider an arbitrary clause $C$. 
   By construction, there is a positive edge $e$ with $L(e)=C$.
   If $c(e)=x_q$ for some variable $x_q$ then by Claim~\ref{claim:one-value-multi} all edges in $c^{-1}(x_q)$ are positive, and hence $f(x_q)=T$.
   Since $c(e)\in L(e)$ we have $x_q \in C$, so $C$ is satisfied.
   If $c(e)=\neg x_q$ for some variable $x_q$ then by Claim~\ref{claim:one-value-multi} all edges in $c^{-1}(x_q)$ are negative and hence $f(x_q)=F$.
   Again, since $c(e)\in L(e)$ we have $\neg x_q \in C$, so $C$ is satisfied.   

   Assume $\varphi$ is satisfiable and let $f:\vrb(\varphi)\rightarrow\{T,F\}$ be a satisfying assignment.
   We define a list edge coloring $c$ of $(G,L)$ as follows.
   Recall that for every $r\in[10]$, and for every clause $C\in\Cc_r$ there is an edge $e_C^1$ with $L(e_C^1)=C$ and edges $e_C^2$, $e_C^3$ with $L(e_C^2)=L(e_C^3)=\{x_i,\neg x_i,x_j,\neg x_j,x_k,\neg x_k\}$, where $x_i$, $x_j$ and $x_k$ are the three variables that appear in $C$.
   We color $e_C^1$ with any of the satisfied literals of $C$.
   By symmetry assume $c(e_C^1)\in\{x_i,\neg x_i\}$.
   Then we color $e_C^2$ with $x_j$ if $f(x_j)=T$ and with $\neg x_j$ otherwise.
   Similarly, we color $e_C^3$ with $x_k$ if $f(x_k)=T$ and with $\neg x_k$ otherwise.
   Each of the remaining positive edges $e$ of $G$ has its list equal $\{x_i,\neg x_i\}$ for some $x_i\in\vrb(\varphi)$.
   We color $e$ with $x_i$ if $f(x_i)=T$ and with $\neg x_i$ otherwise.
   It follows that every positive edge is colored with a satisfied literal.
   Every negative edge $\tilde{e}$ has its list equal to $\{x_i,\neg x_i\}$ for some $x_i\in\vrb(\varphi)$.
   We color $\tilde{e}$ with $x_i$ when $f(x_i)=F$ and with $\neg x_i$ when $f(x_i)=T$.
   It follows that every negative edge is colored with an unsatisfied literal.
   Let us show that $c$ does not color incident edges with the same color.
   Since the lists of parallel negative edges are disjoint, in our coloring there are no parallel negative edges of the same color.
   Assume there are two parallel positive edges of the form $v_{2r}v_{2r+1}$ of the same color $\ell$, for some $r\in[10]$. 
   Then the variable of $\ell$ belongs to a clause in $\Cc_r$, for otherwise there is exactly one edge with endpoints $v_{2r}v_{2r+1}$ and with list containing $\ell$.
   However, since $\Cc_r$ is independent in $G_\varphi$, there is exactly one such clause $C$ in $\Cc_r$. 
   It follows that the two parallel edges are among the three edges $e_C^1, e_C^2, e_C^3$.
   However, these three edges got different colors, a contradiction. 
   If two edges are incident but not parallel, one of them is positive and the other negative.
   The former is colored with a satisfied literal and the latter with an unsatisfied literal, so they are colored differently. 
   Hence $c$ is a proper list edge coloring, as required.
   This ends the proof of Lemma~\ref{lem:reduction-multi}.

Theorem~\ref{th:multi} follows immediately from Lemmas~\ref{lem:tovey} and~\ref{lem:reduction-multi} and the NP-hardness of \probTSAT.

\section{Hardness of \probLEC in Simple Graphs}\label{sect:simple}

This section is devoted to the proof of the following lemma.

\begin{lemma}
\label{lem:reduction-simple}
For any instance $\varphi$ of \probTFSAT with $n$ variables there is an equivalent instance $(G,L)$ of \probLECSim with $O(\sqrt{n})$ vertices.
Moreover, the instance $(G,L)$ can be constructed in polynomial time.
\end{lemma}

\subsection{Intuition} 
 The general idea is to follow the approach of Lemma~\ref{lem:reduction-multi} and replace the edges with multiplicity $O(n)$ with bipartite graphs with $O(\sqrt{n})$ vertices and $O(n)$ edges. It seems that using only one such graph instead of every bunch of parallel edges with common endpoints is not enough to get a simple graph (though it suffices to reduce the multiplicity to three).
 In our construction, for every $r\in[10]$, we replace every two consecutive bunches of parallel edges between $v_{2r}$, $v_{2r+1}$, and $v_{2r+2}$ from the construction in Lemma~\ref{lem:reduction-multi} by seven layers $L_i$, $i=6r+1,\ldots,6r+7$, each of $O(\sqrt{n})$ vertices, with some edges joining both consecutive and non-consecutive layers. 
 The subgraph induced by $\bigcup_{i=6r+1}^{6r+7}L_i$ is called the $r$-th {\em clause verifying gadget} $G_r$.
 (Note that the layers $L_i$ for $i\equiv 1\pmod 6$ are shared between consecutive gadgets.)
 Analogously as in Lemma~\ref{lem:reduction-multi}, the role of $G_r$ is to check whether all clauses in $\Cc_r$ are satisfied. 
 We add  also two additional layers $L_0$ and $L_{62}$ which make some of our arguments simpler.
 
\subsection{Construction} 
 It will be convenient to assume that $\sqrt{n}\in\mathbb{N}$. 
 We do not lose on generality because otherwise we just add $n^+=(\ceil{\sqrt{n}}+1)^2-n$ variables $y_1,y_2,\ldots,y_{n^+}$ and clauses 
 \[\{y_1,y_2,y_3\},\{y_2,y_3,y_4\},\ldots,\{y_{n^+-2},y_{n^+-1},y_{n^+}\}.\]
 Note that $n^+\ge 3$, $n^+\le(\sqrt{n}+2)^2-n=4\sqrt{n}+4$ and $\sqrt{n+n^+}=\ceil{\sqrt{n}}+1\in\mathbb{N}$.
 Hence we added only $O(\sqrt{n})$ variables and clauses, and the resulting formula is still a \probTFSAT instance. 
 
 We begin as in Lemma~\ref{lem:reduction-multi}, by building the graph $G_\varphi$, finding its greedy coloring $g$ which partitions the clause set into 10 color classes $\Cc_r$, $r\in[10]$.
 Let us build the instance $(G,L)$ step by step.
 
 Add two sets of vertices (called {\em layers}) $L_i=\{v^i_j \mid j\in [\sqrt{n}]\}$, $i=0,1$.
 Then add all possible $n$ edges between $L_0$ and $L_1$ forming a complete bipartite graph.
 Map the $n$ variables to the $n$ edges in a $1-1$ way.
 For every $i\in[n]$, set the list of the edge assigned to $x_i$ to $\{x_i,\neg x_i\}$.
 
 The vertex set $V(G)$ contains further 60 layers of vertices $L_i$, $i=\{2,\ldots,61\}$, where $L_{i} = \{v^{i}_j \mid j\in [6\sqrt{n}+3]\}$.
 Finally, $L_{62} = \{v^{62}_j \mid j\in [\sqrt{n}+1]\}$.
 Denote also $L_{-1} = L_{63} = \emptyset$.
 In what follows we add the remaining edges of $G$.
 Whenever we add edges between $L_i$ and $L_{i-1}$, for every $j<i$ all the edges of the output graph between $L_j$ and $L_{j-1}$ are already added.
 We will make sure to keep the following invariants satisfied during the process of construction (note that they hold for the part constructed so far).
 
 \begin{invariant}[\textbf{Uniqueness}]
 \label{inv:unique}
 For every $i\in[62]$, for every variable $x_j\in \vrb(\varphi)$ there is at most one edge $uv\in E(L_i, L_{i+1})$ such that $\{x_j,\neg x_j\} \cap L(uv)\ne \emptyset$. 
 Moreover, after finishing of adding edges between $L_i$ and $L_{i+1}$, there is exactly one such edge.
 \end{invariant}

 Using the notation from Invariant~\ref{inv:unique}, if the edge $uv$ exists, we can denote $v^+_{i,j}=u$ and $v^-_{i+1,j}=v$.
 
 \begin{invariant}[\textbf{Flow}]
 \label{inv:flow}
 For every $i\in\{1,\ldots,62\}$, for every variable $x_j\in \vrb(\varphi)$ we have that $v^-_{i,j}=v^+_{i,j}$, unless $v^-_{i,j}$ or $v^+_{i,j}$ is undefined.
 Moreover, the equality holds after finishing of adding edges between $L_i$ and $L_{i+1}$.
 \end{invariant}
 
 Thanks to Invariant~\ref{inv:flow}, after finishing of adding edges between $L_i$ and $L_{i+1}$, we can just define $v_{i,j}:=v^-_{i,j}=v^+_{i,j}$ for $i\in\{1,\ldots,61\}$. We also put $v_{0,j}=v^+_{0,j}$ and $v_{62,j}=v^-_{62,j}$.
 In our construction we will use some additional colors apart from the literals. However, the following invariant holds. 
 
 \begin{invariant}[\textbf{Lists}]
 \label{inv:lists}
 For every edge $e$ of $G$, the list $L(e)$ contains at least one literal.
 \end{invariant} 

 For every $i\in[62]$ for every vertex $v\in L_i$ let $\deg^-(v)=|E(L_{i-1}),\{v\}|$ and $\deg^+(v)=|E(L_{i+1}),\{v\}|$.
 
 \begin{invariant}[\textbf{Indegrees}]
 \label{inv:indegree}
 For every $i\in[62]$ for vertex $v\in L_i$  we have $\deg^-(v) \le \sqrt{n}$. 
 \end{invariant}
 
 \begin{invariant}[\textbf{Jumping edges}]
 \label{inv:jump}
 For every $i\in[62]$, for vertex $v\in L_i$ there are at most $\sqrt{n}$ edges from $v$ to layers $L_j$ for $j>i+1$.
 \end{invariant}

 By Invariant~\ref{inv:flow} and Invariant~\ref{inv:lists}, for every vertex $v\in V(G)$ it holds that $\deg^+(v)\le \deg^-(v)$. Hence Invariant~\ref{inv:indegree} gives the claim below.
 
 \begin{claim}[\textbf{Outdegrees}]
 \label{cl:outdegree}
 For every $i\in[62]$ for every vertex $v\in L_i$  we have $\deg^+(v) \le \sqrt{n}$. 
 \end{claim}

 Invariants~\ref{inv:unique} and~\ref{inv:lists} immediately imply the following.
 
 \begin{claim}
 \label{cl:n-edges}
 For every $i\in[62]$, we have $|E(L_i,L_{i+1})| \le n$.
 \end{claim} 

 
%
%
%
%
 
 Let us fix $r\in[10]$. We add the edges of the $r$-th {\em clause verifying gadget} $G_r$. 
 Although $G$ is undirected, we will say that an edge $uv$ between $L_i$ and $L_j$ for $i<j$ is {\em from $u$ to $v$} and {\em from $L_i$ to $L_j$}.
 Below we {\em describe} the edges in $G_r$ in the order which is convenient for the exposition.
 However, the {\em algorithm} adds the edges between layers in the left-to-right order, i.e., for $i<j$, edges to $L_i$ are added before edges to $L_j$.
 
 \begin{enumerate}
  \item 
 Edges to $L_{\ell}$ for $\ell=6r+2,6r+4,6r+6$.
 
 For every clause $C\in\Cc_r$ we do the following.
 Let $x_{i_1}, x_{i_2}, x_{i_3}$ be the three different variables that appear in the literals of $C$.
 Let $v_j = v^-_{\ell-1,i_j}$ for $j=1,2,3$.
 Note that vertices $v_1$, $v_2$, $v_3$ need not be distinct.
 By Claim~\ref{cl:outdegree}, $|N(v_j)\cap L_{\ell}|\le \sqrt{n}$ for $j=1,2,3$.
 Let $S=\{v \in L_{\ell} \mid \deg^-(v) = \sqrt{n}\}$.
 By Claim~\ref{cl:n-edges}, there is $|S|\le\sqrt{n}$.
 Hence, for $j=1,2,3$ we have $|L_{\ell} \setminus (N(\{v_j\}) \cup S)| \ge 4\sqrt{n}+3$ and we can pick a vertex $w_j\in L_{\ell}$ that has at most $\sqrt{n}-1$ edges from $L_{\ell-1}$, is not adjacent to $v_j$, and is different than $w_{j'}$ for each $j'<j$. 
 If $\ell=6k+6$ we additionally require that for every $j=1,2,3$, the vertex $w_j$ is not adjacent to $v^-_{6r+2,i_j}$ or $v^-_{6r+4,i_j}$.
 By Invariant~\ref{inv:jump} this eliminates at most $2\sqrt{n}$ more candidates, so it is still possible to choose all the $w_j$'s.
 For each $j=1,2,3$, we add an edge $v_jw_j$ with $L(v_jw_j)=\{x_{i_j},\neg x_{i_j}\}$.
 Moreover, if $\ell=6k+6$, for every $j=1,2,3$ we add 
   an edge $v^-_{6r+2,i_j}w_j$ with list $\{x_{i_j},\neg x_{i_j},a_{i,j}\}$ and 
   an edge $v^-_{6r+4,i_j}w_j$ with list $\{x_{i_j},\neg x_{i_j},b_{i,j}\}$.
 The conditions used to choose $w_1$, $w_2$ and $w_3$ guarantee that we do not introduce parallel edges.
 
 For every variable $x_i$ that is not present in any of the clauses of $\Cc_r$ we find a vertex $w\in L_{\ell}$ that has at most $\sqrt{n}-1$ edges from $L_{\ell-1}$ and is not adjacent to $v^-_{\ell-1,i}$. 
 Again, this is possible because there are at most $2\sqrt{n}$ vertices in $L_{\ell}$ that violate any of these constraints.
 We add an edge $v^-_{\ell-1,i}w$ with $L(v^-_{\ell-1,i}w)=\{x_{i},\neg x_{i}\}$.
 
 Note that all invariants are satisfied: for Invariant~\ref{inv:unique} it follows from the fact that $\Cc_r$ is independent in $G_\varphi$, while invariants~\ref{inv:flow},~\ref{inv:lists},~\ref{inv:indegree} follow immediately from the construction.
 Invariant~\ref{inv:jump} stays satisfied after adding $v^-_{6r+2,i_j}w_j$ because for every variable $x_k$ such that $v^-_{6r+2,k}=v^-_{6r+2,i_j}$ we add at most one edge from $v^-_{6r+2,i_j}$ to $L_{6r+6}$, and the number of such variables is equal to $\deg^-(v^-_{6r+2,i_j})$, which is at most $\sqrt{n}$ by Invariant~\ref{inv:indegree} (analogous argument applies to adding the edge $v^-_{6r+4,i_j}w_j$).

 \item
 Edges to $L_{\ell}$ for $\ell=6r+3,6r+5,6r+7$.
 
 For every clause $C\in\Cc_r$ we do the following.
 Let $C=\{\ell_1,\ell_2,\ell_3\}$ and let $x_{i_j}$ be the variable from the literal of $\ell_j$, for $j=1,2,3$.
 Let $w_j = v^-_{\ell-1,i_j}$ for $j=1,2,3$.
 By Claim~\ref{cl:outdegree}, $|N(\{w_1,w_2,w_3\}\cap L_{\ell})|\le 3\sqrt{n}$.
 Also, there are at most $\sqrt{n}+2$ vertices in $L_{\ell}$ with at least $\sqrt{n}-2$ edges from $L_{\ell-1}$.
 Indeed, otherwise $|E(L_{\ell-1},L_{\ell})|\ge n + \sqrt{n} - 6$ and either $n \le 36$ (and the lemma is trivial) or there is a contradiction with Claim~\ref{cl:n-edges}.
 Hence, we can find a vertex $z_{\ell,C}\in L_{\ell}$ that has at most $\sqrt{n}-3$ edges to $L_{\ell-1}$ and is not adjacent to $\{w_1,w_2,w_3\}$. 
 If $\ell=6k+7$ we additionally require that the vertex $z_{6k+7,C}$ is not adjacent to $z_{6k+3,C}$ or $z_{6k+5,C}$.
 By Invariant~\ref{inv:jump} this eliminates at most $2\sqrt{n}$ more candidates, so it is still possible to choose vertex $z_{6k+7,C}$.
 For each $j=1,2,3$, we add an edge $w_jz_{\ell,C}$.
 We put $L(w_jz_{6r+3,C})=\{x_{i_j},\neg x_{i_j},a_{i_j}\}$, 
        $L(w_jz_{6r+5,C})=\{x_{i_j},\neg x_{i_j},b_{i_j}\}$, and
        $L(w_jz_{6r+7,C})=\{\ell_j,c_{C},d_{C}\}$.
 (The colors $a_{i_j},b_{i_j},c_{C},d_{C}$ are not literals --- these are new auxiliary colors; each variable $x_i$ has its own distinct auxiliary colors $a_i, b_i$, and each clause $C$ has its own auxiliary colors  $c_C, d_C$.)
 We add edges $z_{6r+3,C}z_{6r+7,C}$ and $z_{6r+5,C}z_{6r+7,C}$, both with lists $\{x_{i_1},\neg x_{i_1},x_{i_2},\neg x_{i_2},x_{i_3},\neg x_{i_3}\}$.
 
 For every variable $x_i$ that is not present in any of the clauses of $\Cc_r$ we proceed analogously as in Step 1.
 
 The invariants hold for the similar reasons as before. In particular,  Invariant~\ref{inv:jump} stays satisfied after adding $z_{6r+3,C}z_{6r+7,C}$ because for every clause $C'$ such that $z_{6r+3,C'}=z_{6r+3,C}$ we add exactly one edge from $z_{6r+3,C}$ to $L_{6r+7}$, and the number of such clauses is bounded by $\deg^-(z_{6r+3,C})/3$, which is at most $\sqrt{n}/3$ by Invariant~\ref{inv:indegree} (analogous argument applies to adding the edge $z_{6r+5,C}z_{6r+7,C}$).
%
%
%

 \end{enumerate}
 
 Finally, we add edges between $L_{61}$ and $L_{62}$. 
 For every variable $x_i$ we find a vertex $w\in L_{62}$ that is not adjacent to $v^-_{61,i}$, which is possible because $\deg^+(v^-_{61,i}) \le \sqrt{n}$.
 We add an edge $v^-_{61,i}w$ with $L(v^-_{61,i}w)=\{x_{i},\neg x_{i}\}$.

 The following claims follow directly from the construction.
 
 \begin{claim}
 \label{cl:struct1}
 For every $r\in[10]$, for every clause $C\in\Cc_r$ with variables $x_{i_1}, x_{i_2}, x_{i_3}$, and for each $\ell=6r+3,6r+5,6r+7$ we have $v_{\ell,i_1}=v_{\ell,i_2}=v_{\ell,i_3}=z_{\ell,C}$.
  Moreover, for each $\ell=6r+3,6r+5,6r+7$ and $j=1,2,3$ we have $L(z_{\ell,C}v_{\ell+1,i_j})=\{x_{i_j},\neg x_{i_j}\}$.
 \end{claim}
  
 \begin{claim}
 \label{cl:struct2}
 For every edge $uv\in E(G)$, where $u\in L_j, v\in L_k$, if $\{x_i,\neg x_i\}\cap L(uv)\ne\emptyset$, then $u=v_{j,i}$ and $u=v_{k,i}$.
 \end{claim}
  
  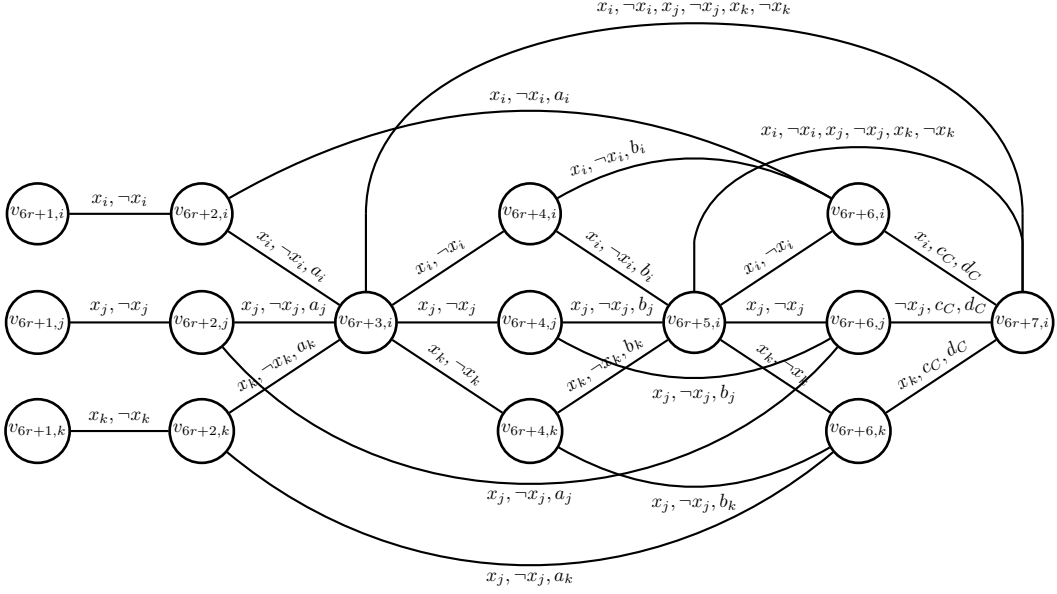
\begin{figure}
  \begin{center}
\begin{tikzpicture}[scale=0.72, transform shape]
\tikzset{vtx/.style={draw, circle, line width=1pt, inner sep = 1pt}}
\tikzset{inv/.style={circle, line width=1pt, inner sep = 1pt}}
\tikzset{edgecolor/.style={draw, rectangle, line width=.5pt, fill=white, inner sep=2pt, above}}

\node [vtx] (v11) at (0,4) {$v_{6r+1,i}$};
\node [vtx] (v12) at (0,2) {$v_{6r+1,j}$};
\node [vtx] (v13) at (0,0) {$v_{6r+1,k}$};

\node [vtx] (v21) at (3,4) {$v_{6r+2,i}$};
\node [vtx] (v22) at (3,2) {$v_{6r+2,j}$};
\node [vtx] (v23) at (3,0) {$v_{6r+2,k}$};

\node [vtx] (v3) at (6,2) {$v_{6r+3,i}$};

\node [vtx] (v41) at (9,4) {$v_{6r+4,i}$};
\node [vtx] (v42) at (9,2) {$v_{6r+4,j}$};
\node [vtx] (v43) at (9,0) {$v_{6r+4,k}$};

\node [vtx] (v5) at (12,2) {$v_{6r+5,i}$};

\node [vtx] (v61) at (15,4) {$v_{6r+6,i}$};
\node [vtx] (v62) at (15,2) {$v_{6r+6,j}$};
\node [vtx] (v63) at (15,0) {$v_{6r+6,k}$};

\node [vtx] (v7) at (18,2) {$v_{6r+7,i}$};

\draw [thick] (v11) edge node [above,pos=0.5] {$x_i,\neg x_i$} (v21);
\draw [thick] (v12) edge node [above,pos=0.5] {$x_j,\neg x_j$} (v22);
\draw [thick] (v13) edge node [above,pos=0.5] {$x_k,\neg x_k$} (v23);

\draw [thick] (v21) edge node [above,pos=0.5,sloped] {$x_i,\neg x_i,a_i$} (v3);
\draw [thick] (v22) edge node [above,pos=0.5] {$x_j,\neg x_j,a_j$} (v3);
\draw [thick] (v23) edge node [above,pos=0.5,sloped] {$x_k,\neg x_k,a_k$} (v3);

\draw [thick] (v3) edge node [above,pos=0.5,sloped] {$x_i,\neg x_i$} (v41);
\draw [thick] (v3) edge node [above,pos=0.5] {$x_j,\neg x_j$} (v42);
\draw [thick] (v3) edge node [above,pos=0.5,sloped] {$x_k,\neg x_k$} (v43);

\draw [thick] (v41) edge node [above,pos=0.5,sloped] {$x_i,\neg x_i,b_i$} (v5);
\draw [thick] (v42) edge node [above,pos=0.5] {$x_j,\neg x_j,b_j$} (v5);
\draw [thick] (v43) edge node [above,pos=0.5,sloped] {$x_k,\neg x_k,b_k$} (v5);

\draw [thick] (v5) edge node [above,pos=0.5,sloped] {$x_i,\neg x_i$} (v61);
\draw [thick] (v5) edge node [above,pos=0.5] {$x_j,\neg x_j$} (v62);
\draw [thick] (v5) edge node [above,pos=0.5,sloped] {$x_k,\neg x_k$} (v63);

\draw [thick] (v61) edge node [above,pos=0.5,sloped] {$x_i,c_C,d_C$} (v7);
\draw [thick] (v62) edge node [above,pos=0.5] {$\neg x_j,c_C,d_C$} (v7);
\draw [thick] (v63) edge node [above,pos=0.5,sloped] {$x_k,c_C,d_C$} (v7);

\draw [thick] (v3) -- +(0,2cm) edge [out=90,in=90] node [above,sloped] {$x_i,\neg x_i,x_j,\neg x_j,x_k,\neg x_k$} ([yshift=2cm]v7);

\draw [thick] (v7) -- +(0,2cm);

\draw [thick] (v5) -- +(0,1.5cm) edge [out=80,in=100] node [above] {$x_i,\neg x_i,x_j,\neg x_j,x_k,\neg x_k$} ([yshift=1.5cm]v7);
\draw [thick] (v7) -- +(0,1.5cm);

\draw [thick] (v21) edge [bend left] node [above] {$x_i,\neg x_i,a_i$} (v61);
\draw [thick] (v41) edge [bend left] node [above,pos=0.2,sloped] {$x_i,\neg x_i,b_i$} (v61);

\draw [thick] (v22) edge [bend right=50] node [below] {$x_j,\neg x_j,a_j$} (v62);
\draw [thick] (v42) edge [bend right] node [below] {$x_j,\neg x_j,b_j$} (v62);

\draw [thick] (v23) edge [bend right=40] node [below] {$x_j,\neg x_j,a_k$} (v63);
\draw [thick] (v43) edge [bend right] node [below] {$x_j,\neg x_j,b_k$}(v63);

\end{tikzpicture}
    \end{center}
\caption{\label{fig:gadget}Edges in the gadget $G_r$ related to a clause $(x_i \vee \neg x_j \vee x_k)$ from $\Cc_r$.}
\end{figure}

  This finishes the description of the output instance. 
 Since $G$ contains $O(1)$ layers, each with $O(\sqrt{n})$ vertices, it follows that $|V(G)|=O(\sqrt{n})$, as required.
 See Fig~\ref{fig:gadget} for an illustration of edges representing a single clause within a clause verifying gadget.
 
\subsection{Structure of coloring}  
Similarly as for multigraphs the crux of the equivalence between instances is the following claim.
 
   \begin{claim} \label{claim:variable-path-simple}
    For every list edge coloring $c$ of $(G,L)$, for every $i\in [n]$, the edges in $c^{-1}(\{x_i,\neg x_i\})$ form a path $P_i$ from $L_0$ to $L_{62}$.
    Moreover, if $P_i$ contains an edge $v_{6r+6,i}v_{6r+7,i}$ for some $r\in[10]$, then this edge in preceded by an even number of edges on $P_i$.
   \end{claim}

   \begin{proof}
   Fix $i\in[n]$. For convenience, denote $E_i=c^{-1}(\{x_i,\neg x_i\})$. 
   By Invariant~\ref{inv:unique} there is exactly one edge between $L_0$ and $L_1$ that has $x_i$ or $\neg x_i$ on its list, namely $v_{0,i}v_{1,i}$.
   Similarly, there is exactly one edge between $L_{61}$ and $L_{62}$ that has $x_i$ or $\neg x_i$ on its list, namely $v_{61,i}v_{62,i}$.
   Since $L(v_{0,i}v_{1,i})=L(v_{61,i}v_{62,i})=\{x_i,\neg x_i\}$, we know that $v_{0,i}v_{1,i},v_{61,i}v_{62,i}\in E_i$, and these are the only edges of $E_i$ in $E(L_0, L_1) \cup E(L_{61}, L_{62})$.
   Observe that edges between non-consecutive layers never leave the clause verifying gadgets. 
   Hence, for the first part of the claim, it suffices to show that for every $r\in[10]$, the edges in $E_i \cap E(G_r)$ form a path between $v_{6r+1,i}$ and $v_{6r+7,i}$. 
   In fact, by Claim~\ref{cl:struct2} it suffices to show that $E_i \cap E(G_r)$ {\em contains} a path between $v_{6r+1,i}$ and $v_{6r+7,i}$ that visits all the vertices $\{v_{6r+j,i}\mid j=1,\ldots,7\}$.
   To this end, fix $r\in[10]$.

   First assume that $x_i$ does not appear in any clause of $\Cc_r$.
   Then $G_r$ contains the path $v_{6r+1,i},v_{6r+2,i},\ldots, v_{6r+7,i}$, where each edge has the list $\{x_i,\neg x_i\}$.
   It immediately implies that all edges of this path are in $E_i \cap E(G_r)$, as required.

   Now let us assume that $x_i$ appears in a clause $C\in\Cc_r$.
   Let $C=\{\ell_i, \ell_j, \ell_k\}$ and assume that the literal $\ell_i$ contains $x_i$, the literal $\ell_j$ contains a variable $x_j$, and the literal $\ell_k$ contains a variable $x_k$.
   Observe that for $j=1,3,5$ we have $v_{6r+j,i}v_{6r+j+1,i} \in E_i$ because these edges have their lists equal to $\{x_i,\neg x_i\}$.
   Note also that $\Delta(E_i)\le 2$ because $E_i$ is a union of two matchings (colors).
   We consider three subcases.
   
   \begin{enumerate}
    \item Assume $v_{6r+3,i}v_{6r+7,i} \in E_i$.
          Since $\Delta(E_i)\le 2$ and $v_{6r+3,i}v_{6r+4,i} \in E_i$ we know that $v_{6r+2,i}v_{6r+3,i}\not\in E_i$, and as a consequence, $c(v_{6r+2,i}v_{6r+3,i})=a_{i}$. Hence $c(v_{6r+2,i}v_{6r+6,i})\ne a_{i}$, which implies that $v_{6r+2,i}v_{6r+6,i}\in E_i$.
          Then, since $\Delta(E_i)\le 2$ and $v_{6r+5,i}v_{6r+6,i} \in E_i$ we know that $v_{6r+4,i}v_{6r+6,i}\not\in E_i$, and as a consequence, $c(v_{6r+4,i}v_{6r+6,i})=b_{i}$. Hence $c(v_{6r+4,i}v_{6r+5,i})\ne b_{i}$, which implies that $v_{6r+4,i}v_{6r+5,i}\in E_i$.
          Thus, we have shown that $E_i$ contains the path $v_{6r+1,i},v_{6r+2,i},v_{6r+6,i},v_{6r+5,i},v_{6r+4,i},v_{6r+3,i},v_{6r+7,i}$, as required.
    \item Assume $v_{6r+5,i}v_{6r+7,i} \in E_i$.
          Since $\Delta(E_i)\le 2$ and $v_{6r+5,i}v_{6r+6,i} \in E_i$ we know that $v_{6r+4,i}v_{6r+5,i}\not\in E_i$, and as a consequence, $c(v_{6r+4,i}v_{6r+5,i})=b_{i}$. Hence $c(v_{6r+4,i}v_{6r+6,i})\ne b_{i}$, which implies that $v_{6r+4,i}v_{6r+6,i}\in E_i$.
          Then, since $\Delta(E_i)\le 2$ and $v_{6r+5,i}v_{6r+6,i} \in E_i$ we know that $v_{6r+2,i}v_{6r+6,i}\not\in E_i$, and as a consequence, $c(v_{6r+2,i}v_{6r+6,i})=a_{i}$. Hence $c(v_{6r+2,i}v_{6r+3,i})\ne a_{i}$, which implies that $v_{6r+2,i}v_{6r+3,i}\in E_i$.
          Thus, we have shown that $E_i$ contains the path $v_{6r+1,i},v_{6r+2,i},v_{6r+3,i},v_{6r+4,i},v_{6r+6,i},v_{6r+5,i},v_{6r+7,i}$, as required.
    \item Assume $v_{6r+3,i}v_{6r+7,i},v_{6r+5,i}v_{6r+7,i} \not\in E_i$.
          Since $L(v_{6r+3,i}v_{6r+7,i}) = L(v_{6r+5,i}v_{6r+7,i}) = \{x_i, \neg x_i, x_j, \neg x_j, x_k, \neg x_k\}$
	  we infer that $v_{6r+3,i}v_{6r+7,i},$ $v_{6r+5,i}v_{6r+7,i} \in E_j \cup E_k$.
          By Claim~\ref{cl:struct1} we know that $v_{6r+7,i}=v_{6r+7,j}=v_{6r+7,k}$, $v_{6r+7,i}v_{6r+8,j}\in E_j$ and $v_{6r+7,i}v_{6r+8,k}\in E_k$.
          Since $\Delta(E_j)\le 2$ and $\Delta(E_k)\le 2$, we get that $v_{6r+3,i}v_{6r+7,i}\in E_j$ and $v_{6r+5,i}v_{6r+7,i} \in E_k$ or vice versa.
          In any case, $v_{6k+6,j},v_{6k+7,i}\not\in E_j$, and $v_{6k+6,k},v_{6k+7,i}\not\in E_k$.
          Recall that $L(v_{6k+6,j},v_{6k+7,i}) = \{\ell_j,c_C,d_C\}$ and $L(v_{6k+6,k},v_{6k+7,i}) = \{\ell_k,c_C,d_C\}$.
          It follows that $c(\{v_{6k+6,j}v_{6k+7,i}, v_{6k+6,k}v_{6k+7,i}\})=\{c_C,d_C\}$.
          Then $c(v_{6k+6,i},v_{6k+7,i})\not\in\{c_C,d_C\}$.
          Since $L(v_{6k+6,i},v_{6k+7,i}) = \{\ell_i,c_C,d_C\}$, we get that $v_{6k+6,i},v_{6k+7,i} \in E_i$.
          Then, since $\Delta(E_i)\le 2$ and $v_{6r+5,i}v_{6r+6,i} \in E_i$ we know that $v_{6r+2,i}v_{6r+6,i},v_{6r+4,i}v_{6r+6,i}\not\in E_i$, and as a consequence, $c(v_{6r+2,i}v_{6r+6,i})=a_{i}$ and $c(v_{6r+4,i}v_{6r+6,i})=b_{i}$. 
          Hence $c(v_{6r+2,i}v_{6r+3,i})\ne a_{i}$, and $c(v_{6r+4,i}v_{6r+5,i})\ne b_{i}$
	  which implies that $v_{6r+2,i}v_{6r+3,i},v_{6r+4,i}v_{6r+5,i}$ $\in E_i$.
          Thus, we have shown that $E_i$ contains the path $v_{6r+1,i},v_{6r+2,i},v_{6r+3,i},$ $v_{6r+4,i},v_{6r+5,i},v_{6r+6,i},v_{6r+7,i}$, as required.
   \end{enumerate}
   For the second part of the claim recall that $P_i$ decomposes into an edge from $L_0$ to $L_1$, 10 paths of length 6 inside the gadgets and an edge from $L_{61}$ to $L_{62}$. Moreover, if $P_i$ contains an edge $v_{6r+6,i}v_{6r+7,i}$ for some $r\in[10]$, then this edge is the last edge of one of the 10 paths of length 6. It follows that it is preceded by $1+6r+5$ edges, which is an even number. 
      \end{proof}

 \subsection{Equivalence}  
 
   Assume $c$ is a list edge coloring of $(G,L)$. Define a boolean assignment $f:\vrb(\varphi)\rightarrow\{T,F\}$ by setting $x_i$ to $T$ iff the first edge of the path $P_i$ from Claim~\ref{claim:variable-path-simple} is colored by $x_i$. 
   Note that $P_i$ is colored alternately with $x_i$ and $\neg x_i$ and every odd edge on $P_i$ (i.e., preceded by an even number of edges) is colored with a satisfied literal.
   Now consider an arbitrary clause $C$. Let $r=g(C)$.
   Let $C=\{\ell_1,\ell_2,\ell_3\}$ and let $x_{i_j}$ be the variable from the literal of $\ell_j$, for $j=1,2,3$.
   By construction, there are three edges $v_{6r+6,i_j}z_{6r+7,C}$, for $j=1,2,3$ with $L(v_{6r+6,i_j}z_{6r+7,C})=\{\ell_j,c_C, d_C\}$.
   At most two of these edges are colored with $c_C$ or $d_C$, so there is $j=1,2,3$ such that $c(v_{6r+6,i_j}z_{6r+7,C}) = \ell_j$.
   In particular, $v_{6r+6,i_j}z_{6r+7,C} \in c^{-1}(\{x_{i_j},\neg x_{i_j}\})$ and hence, by Claim~\ref{claim:variable-path-simple} we know that $v_{6r+6,i_j}z_{6r+7,C}\in P_{i_j}$. 
   However, by the second part of Claim~\ref{claim:variable-path-simple} this edge is preceded by an even number of edges on $P_{i_j}$. 
   It follows that $\ell_j$ is satisfied.

   Assume $\varphi$ is satisfiable and let $f:\vrb(\varphi)\rightarrow\{T,F\}$ be a satisfying assignment.
   We define a list edge coloring $c$ of $(G,L)$ as follows.
   Consider any edge $e\in E(L_0,L_1)$. 
   Then $L(e)=\{x_i,\neg x_i\}$.
   We color $e$ with $x_i$ when $f(x_i)=T$ and with $\neg x_i$ otherwise.
   Now consider any edge $e\in E(L_{61},L_{62})$. 
   Again $L(e)=\{x_i,\neg x_i\}$.
   We color $e$ with $x_i$ when $f(x_i)=F$ and with $\neg x_i$ otherwise.
   By Invariant~\ref{inv:unique} incident edges get different colors in the partial coloring described so far.
   In what follows we describe $c|_{E(G_r)}$ for every $r\in[10]$ separately.
   Fix $r\in [10]$.
   
   Consider an arbitrary clause $C\in\Cc_r$.
   Let $C=\{\ell_1,\ell_2,\ell_3\}$ and let $x_{i_j}$ be the variable from the literal of $\ell_j$, for $j=1,2,3$.
   Since $\varphi$ is satisfied by $f$, at least one literal of $C$ is satisfied by $f$, by symmetry we can assume it is $\ell_1$.
   Consider the three edge disjoint paths
   \begin{align*}
      R_1 &=v_{6r+1,i_1},v_{6r+2,i_1},v_{6r+3,i_1}, v_{6r+4,i_1}, v_{6r+5,i_1},v_{6r+6,i_1},v_{6r+7,i_1}, \\
      R_2 & =v_{6r+1,i_2},v_{6r+2,i_2},v_{6r+6,i_2},v_{6r+5,i_2},v_{6r+4,i_2},v_{6r+3,i_2},v_{6r+7,i_2}, \\
      R_3 &= v_{6r+1,i_3},v_{6r+2,i_3},v_{6r+3,i_3},v_{6r+4,i_3},v_{6r+6,i_3},v_{6r+5,i_3},v_{6r+7,i_3}.
    \end{align*}
   For each $j=1,2,3$ the path $R_j$ is colored by $x_{i_j}$ and $\neg x_{i_j}$ alternately, beginning with $\neg x_{i_j}$ if $f(x_{i_j})=T$ and with $x_{i_j}$ if $f(x_{i_j})=F$.
   Note that edges of $R_1$, $R_2$ and $R_3$ are colored by colors from their lists. 
   Indeed, this is obvious for every edge apart from $v_{6r+6,i_1},v_{6r+7,i_1}$, because their lists contain $\{x_{i_j},\neg x_{i_j}\}$.
   Edge $v_{6r+6,i_1},v_{6r+7,i_1}$ is colored with $x_{i_j}$ if $f(x_{i_j})=T$ and with $\neg x_{i_j}$ if $f(x_{i_j})=F$.
   It follows that $v_{6r+6,i_1},v_{6r+7,i_1}$ is colored with the literal from $\{x_{i_1},\neg x_{i_1}\}$ which is satisfied by $f$, hence it is colored by $\ell_1$, and $\ell_1\in L(v_{6r+6,i_1},v_{6r+7,i_1})$, as required.   
   Finally, we put $c(v_{6r+2,i_1}v_{6r+6,i_1})=a_{i_1}$, $c(v_{6r+4,i_1}v_{6r+6,i_1})=b_{i_1}$, 
                   $c(v_{6r+2,i_2}v_{6r+3,i_2})=a_{i_2}$, $c(v_{6r+4,i_2}v_{6r+6,i_2})=b_{i_2}$,
                   $c(v_{6r+2,i_3}v_{6r+6,i_3})=a_{i_3}$, $c(v_{6r+4,i_3}v_{6r+5,i_3})=b_{i_3}$,
                   $c(v_{6r+6,i_2}v_{6r+7,i_2})=c_C$, 
                   $c(v_{6r+6,i_3}v_{6r+7,i_3})=d_C$.
   Thus we have colored all edges of $G_r$ which have lists containing a variable from $C$.
   
   Now consider any variable $x_i$ that does not appear in any clause of $\Cc_r$.
   Consider the path $v_{6r+1,i},v_{6r+2,i},\ldots, v_{6r+7,i}$.
   If $f(x_i)=T$, color the path with the sequence of colors $\neg x_i,x_i,\neg x_i,\ldots, x_i$, and otherwise with the sequence of colors $x_i,\neg x_i,x_i,\ldots, \neg x_i$. 
   
   Thus we have colored all the edges of $G_r$. It is straightforward to check that for every $r\in[10]$ the subgraph $G_r$ is colored properly.
   It remains to show that vertices in the layers $L_i$ for $i\equiv 1 \pmod 6$ are not incident to two edges of the same color.
   Clearly, this cannot happen for colors $a_j$ or $b_j$ for any $j\in[n]$, because they are not present on lists of edges incident to $L_i$ for $i\equiv 1 \pmod 6$.
   Also, it cannot happen for colors $c_C$ or $d_C$ for any clause $C$, because edges with these colors on their list only join $L_{i-1}$ with $L_{i}$ for $i\equiv 1 \pmod 6$, so two incident edges colored with $c_C$ or $d_C$ cannot belong to different gadgets.
   Finally, consider colors $\{x_i,\neg x_i\}$ for a fixed $i\in [n]$.
   The edges with these colors form a path of length 62, starting with $v_{0,i}v_{1,i}$, and continued as follows.
   The edge $v_{0,i}v_{1,i}$ is followed by 10 paths of length 6.
   For every $r\in[10]$, the $r$-th path of length 10 begins in $v_{6r+1,i}$ and ends in $v_{6r+7,i}=v_{6(r+1)+1,i}$.
   Finally, the 62-path ends with edge $v_{61,i}v_{62,i}$.
   Note that $v_{0,i}v_{1,i}$ is colored with the satisfied literal.
   Next, for every $r\in[10]$, the first edge of the $r$-th 10-path is colored with the non-satisfied literal and its last edge is colored by the satisfied literal.
   Finally, $v_{61,i}v_{62,i}$ is colored with the non-satisfied literal. 
   It follows that the 62-path of all edges with colors from $\{x_i,\neg x_i\}$ is colored alternately in $x_i$ and $\neg x_i$, as required.
   This finishes the proof that $c$ is a list edge coloring of $(G,L)$, and the proof of Lemma~\ref{lem:reduction-simple}.

\subsection{Proof of Theorem~\ref{th:simple}}   
Theorem~\ref{th:simple} follows immediately from Lemma~\ref{lem:reduction-simple} and Corollary~\ref{cor:eth-3,4-sat}.
Indeed, if there is an algorithm $A$ which solves \probLECSim in time $2^{o(|V(G)|^2)}$, then by Lemma~\ref{lem:reduction-simple} an $n$-variable instance of \probTFSAT can be transformed to a $O(\sqrt{n})$-vertex instance of \probLECSim in polynomial time and next solved in time $2^{o(n)}$ using $A$, which contradicts ETH by Corollary~\ref{cor:eth-3,4-sat}.

\section{Conclusions and further research}

In this work we have  shown that \probLECSim does not admit an algorithm in time $2^{o(n^2)}$, unless ETH fails.
This has consequences for designing algorithms for \probEC: in order to break the barrier $2^{O(n^2)}$ one has to use methods that exploit symmetries between colors, and in particular do not apply to the list version.
On the other hand, one may hope that our reductions can inspire a reduction to \probEC which would exclude at least a $2^{O(n)}$-time algorithm.
However it seems that \probEC requires a significantly different approach.
In our reductions we were able to encode information (namely, the boolean value of a variable in a satisfying assignment) in a {\em color} of an edge.
In the case of \probEC this is not possible, because one can recolor any edge $e$ by choosing an arbitrary different color $c'$ and swapping $c'$ and the color $c$ of $e$ on the maximal path/cycle that contains $e$ and has edges colored with $c$ and $c'$ only.

\bibliographystyle{abbrv}
\bibliography{listedgecol}

\end{document}